%% file: 0main.tex
\documentclass{sig-alternate}

\usepackage{stmaryrd}
\usepackage{cite}
\usepackage{paralist}
\usepackage{amsmath}
\usepackage{amsthm}
\usepackage{amsfonts}
\usepackage{graphicx}
\usepackage{epstopdf}
\usepackage{subcaption}
\usepackage{subcaption}
\usepackage{listings}
\usepackage{paralist}
\usepackage{tikz}
\usepackage[ruled,commentsnumbered,linesnumbered]{algorithm2e}
\usepackage{lipsum}
\usepackage[hidelinks]{hyperref}
\usepackage{url}
\usepackage{multicol}
\usepackage{blindtext}
\usepackage{array}
\usepackage{floatrow}
\usepackage{thmtools}
\usepackage{caption}
\usepackage{lipsum}
\usepackage{float}
\usepackage{arydshln}
\usepackage{multirow}
\usepackage{hhline}

\hypersetup{
    bookmarks=true,         
    unicode=false,          
    pdftoolbar=true,        
    pdfmenubar=true,        
    pdffitwindow=false,     
    pdfstartview={FitH},    
    pdftitle={My title},    
    pdfauthor={Author},     
    pdfsubject={Subject},   
    pdfcreator={Creator},   
    pdfproducer={Producer}, 
    pdfkeywords={keyword1} {key2} {key3}, 
    pdfnewwindow=true,      
    colorlinks=true,       
    linkcolor=blue,          
    citecolor=blue,        
    filecolor=magenta,      
    urlcolor=blue           
}

\newfont{\mycrnotice}{ptmr8t at 7pt}
\newfont{\myconfname}{ptmri8t at 7pt}
\let\crnotice\mycrnotice%
\let\confname\myconfname%

\input{prelude1}

\newcommand{\safe}[0]{\mathit{Safe}}
\newcommand{\unsafe}[0]{\mathit{Unsafe}}
\newcommand{\init}[0]{\mathit{Init}}

\newcommand{\goal}[0]{\mathit{Goal}}

\newcommand{\advlev}{adversary leverage}
\newcommand{\bloat}{initialization factor}
\newcommand{\strength}{strengthening}
\newcommand{\ARAS}[0]{\mathit{ARAC}}

\newtheorem{definition}{Definition}

\newcounter{Theorem}
\newtheorem{lemma}[Theorem]{Lemma}
\newtheorem{corollary}[Theorem]{Corollary}
\newtheorem{theorem}[Theorem]{Theorem}

\begin{document}

\title{Controller Synthesis for Linear Time-varying Systems with Adversaries}

\numberofauthors{4}
\author{
Zhenqi Huang \and Yu Wang \and Sayan Mitra \and Geir Dullerud 
\end{tabular} 
\\
\begin{tabular}{c}
\affaddr{\{zhuang25, yuwang8, mitras, dullerudge\}@illinois.edu}\\
\affaddr{Coordinate Science Laboratory}\\
\affaddr{University of Illinois at Urbana Champaign}\\
\affaddr{Urbana, IL 61801} \\
}

\maketitle

\begin{abstract}
We present a controller synthesis algorithm for a discrete time
reach-avoid problem in the presence of adversaries. Our model of the
adversary captures typical malicious attacks envisioned on cyber-physical
systems such as sensor spoofing, controller corruption, and actuator intrusion.
After formulating the problem in a general setting, we present
a sound and complete algorithm for the case with linear dynamics and an
adversary with a budget on the total L2-norm of its actions. The
algorithm relies on a result from linear control theory that enables us to
decompose and precisely compute the reachable states of the system in
terms of a symbolic simulation of the adversary-free dynamics and the
total uncertainty induced by the adversary. With this decomposition, the
synthesis problem eliminates the universal quantifier on the adversary's
choices and the symbolic controller actions can be effectively solved
using an SMT solver. The constraints induced by the adversary are
computed by solving second-order cone programmings. 
The algorithm is later extended to synthesize state-dependent controller 
 and to generate attacks for the adversary. 
We present
preliminary experimental results that show the effectiveness of this
approach on several example problems.
\end{abstract}

\input{1intro}
\input{related}

\input{3problem}
\input{4decouple}
\input{5algorithm}

\input{6general}

\input{7experiment}

\input{8conclusion}

\bibliographystyle{IEEEtran}
\bibliography{sayan1,HSCC}

\end{document}

%% file: 1intro.tex
\section{Introduction}
\label{sec:intro}

We study a discrete time synthesis problem for a plant simultaneously acted-upon by a {\em controller\/} and an {\em adversary\/}. Synthesizing controller strategies for stabilization in the face of random noise or disturbances is one of the classical problem in control theory~\cite{1995robust,basar1995dynamic}. Synthesis for temporal logic specifications~\cite{tabuada2003model,ulusoy2014incremental,WolffTM14}, for discrete, continuous, and hybrid systems have been studied in detail. The reach-avoid properties that our controllers target are special, bounded-time temporal logic requirements, and they have received special attention as well~\cite{zhou2012general}. Unlike the existing models in controller synthesis literature, however, the system here is afflicted by an adversary and we would like to synthesize a controller that guarantees its safety and liveness for {\em all \/} possible choices made by the adversary.

This problem is motivated by the urgent social  to secure control modules in critical infrastructures and safety-critical systems against malicious attacks~\cite{cardenas2008research,bullo10attack}. Common modes of attack include sensor spoofing or jamming, malicious code, and  actuator intrusion. Abstracting the mechanisms used to launch the attacks, their effect on physical plant can be captured as a switched system with inputs from the controller and the adversary:
\[
x_{t+1} = f_{\sigma_t}(x_t, u_t, a_t),
\]
where $x_t$ is the state of the system, $u_t$ and $a_t$ are the inputs from the controller and the adversary. The problem is parameterized by a family of dynamical functions $\{f_\sigma\}_{\sigma\in \Sigma}$,
a  switching signal $\{\sigma_t\}_{t\in \naturals}$, 
a time bound $T$, the set of initial sates ($\init$), target states ($\goal$), safe states ($\safe$), the set of choices available to the adversary ($Adv$) and the controller ($Ctr$). A natural decision problem is to ask: Does there exist a controller strategy $\vu \in Ctr$ such that for any initial state in $\init$, and any choice by the adversary in $Adv$ the system remains  $\safe$ and reaches $\goal$ within time $T$.
A constructive affirmative answer can be used to implement controllers that are $Adv$-resilient, while a negative answer can inform the system design choices that influence the other parameters like $f$, $T$ and $Ctr$. 

We provide a decision procedure
for this problem for the special case where $f$ is a linear mapping, the sets $\init$, $\safe$, $\goal$, and $Ctr$ sets are given as by polytopic sets and $Adv$ is given as an $\ell^2$ ball in an Euclidean space. The idea behind the algorithm is a novel decomposition that distinguishes it from the LTL-based synthesis approaches~\cite{tabuada2003model} and reachability-based techniques of~\cite{zhou2012general}. The key to this decomposition is the concept of {\em adversarial leverage\/}: the uncertainty in the state of the system induced by the sequence of choices made by the adversary, for a given initial state and a sequence  of choices made by the controller.
For linear models, we show that the \advlev{} can be computed exactly. 
As a result, an adversary-free synthesis problem with a modified set of
$\safe$ and $\goal$ requirements, precisely gives the solution for the problem with adversary.

We implement the algorithm with a convex optimization package 
CVXOPT~\cite{dahl2006cvxopt} and an SMT solver Z3~\cite{de2008z3}.
We present experimental results that show the effectiveness of this
approach on several example problems. 
The algorithm synthesizes adversary-resilient control 
for systems with up to 16 dimensions in minutes.
We have that the  algorithm can be applied to to analyze the maximum 
power of the adversary such that a feasible solution exists  
and to synthesize attacks for adversary.


\paragraph{Advancing Science of Security}

Scientific security analysis is necessarily parameterized by the the skill
and effort level of the adversary. In this paper we combine these
parameters into a single parameter called the {\em budget\/} of the
adversary which can model sensor attacks and actuator intrusions with
different strengths and persistence. We present the foundations for
analyzing cyberphysical systems under attack from these adversaries with
different budgets. Specifically, we develop algorithms for both automatic
synthesis of safe controllers and for proving that there exists no
satisfactory controller, when the adversary has a certain budget. These
algorithms can be also used to characterize vulnerability of system states
in terms of the adversary budget that make them infeasible for safe
control. In summary, we present a framework for algorithmically studying
security of cyberphysical systems in the context of model-based
development.






%% file: related.tex
\section{Related Work}
\label{sec:related}

In this work, we employ SMT solvers to synthesize controllers for reach-avoid problems for discrete-time linear systems with adversaries. 
Our problem is formulated along the line of the framework and fundamental design goals of~\cite{cardenas2008research,cardenas2009challenges}.
The framework was applied to study optimal control design with respect a given objective function under security constraints~\cite{amin2009safe} and the detection of computer attacks with the knowledge of the physical system~\cite{cardenas2011attacks}. 
Similar frameworks were adopted in~\cite{Fawzi2014} where the authors proposed an effective algorithm to estimate the system states and designed feedback controllers to stabilize the system under adversaries, and in~\cite{shoukry2013minimax} where a optimal controller is designed for a distributed control system with communication delays.
Although   the motivation of the above studies are similar to ours, we focus on another aspect of the problem which is to synthesize attack-resilient control automatically.

The idea of using SMT solvers to synthesize feedback controllers for control systems is inspired by recent works~\cite{nedunurismt,beyene2014constraint}. In~\cite{nedunurismt}, the authors used SMT solvers to synthesize integrated task and motion plans by constructing a placement graph. In~\cite{beyene2014constraint}, a constraint-based approach was developed to solve games on infinite graphs between the system and the adversary.
Our work extend the idea of  constraint-based synthesis by introducing  control theoretic approaches
to derived the constraints. 

The authors of~\cite{zhou2012general,ding2011reachability} proposed a game theoretical approach to synthesize controller for the reach-avoid problem, first for continuous and later for switched systems. In these approaches, the reach set of the system is computed by solving a non-linear Hamilton-Jacobi-Isaacs PDE. 
Our methodology, instead of formulating a general optimization problem for which the solution may not be easily computable,
solves a special case exactly and efficiently.
With this building block, we are able to solve more general problems through abstraction and refinement.

%% file: 3problem.tex
\section{Problem Statement}

In this paper, we focus on discrete linear time varying (LTV) systems.
Consider the discrete type linear control system evolving according to the equation:
\begin{equation}
\label{eq:dyn}
x_{t+1} = A_tx_t+B_tu_t+C_ta_t,
\end{equation}
where for each time instant $t\in \naturals$, 
$x_t\in \X \subseteq \reals^n$ is the state vector of the controlled plant, 
$u_t\in \U \subseteq \reals^m$ is controller input to the plant, and 
$a_t\in \A \subseteq \reals^l$ is adversarial input to the plant. 
For a fixed time horizon $T\in \naturals$, let us denote sequences of controller and adversary inputs by $\vu \in \U^T$ and $\va\in \A^T$. 
In addition to the sequence of matrices $A_t$, $B_t$, $C_t$, and a time bound $T$, 
the {\em linear adversarial reach-avoid control\/} problem or $\ARAS$ in short is parameterized by:
\begin{inparaenum}[(i)]
\item three sets of states $\init, \safe, \goal \subseteq \X$ called the {\em initial, safe} and {\em goal\/} states, 
\item a set  $Ctr \subseteq \U^T$ called the {\em controller constraints\/}, and 
\item a set  $Adv \subseteq \A^T$ called the {\em adversary constraints\/}.
\end{inparaenum}
We will assume finite representations of these sets such as polytopes 
and we will state these representational assumptions explicitly later.
A controller input sequence $\vu$ is {\em admissible\/} if it meets the constraints $Ctr$, that is, $\vu \in Ctr$, and  
a adversarial  input sequence  is {\em admissible\/} if $\va \in Adv$. We define what is means to solve a $\ARAS$ problem with an open loop controller strategy.
\begin{definition}
\label{sec:def}
A {\em solution\/} to a $\ARAS$ is an input sequence $\vu \in Ctr$ such that 
for any  initial state $x \in \init$ and any admissible sequence of adversarial inputs $\va \in Adv$, 
the states visited by the system satisfies the condition:
\begin{itemize}
\item {\em (Safe)\/} for all $t \in \{0,\ldots, T\}$, $x_t \in \safe$ and
\item {\em (Winning)\/} $x_T \in \goal$.
\end{itemize} 
\end{definition}
In this paper we propose an algorithm that given a $\ARAS$ problem, either computes its solution or proves that there is none.
In the next section, we discuss how the problem captures instances of control synthesis problems for cyberphysical systems under several different types of attacks.
%


\subsection*{Helicopter Autopilot Example}
\label{sec:heli}
To make this discussion concrete we consider an autonomous helicopter.
The state vector of the plant $x \in \reals^{16}$;
the control input vector $u\in \reals^4$ with bounded range of each component.
The descriptions of the state and input vectors are in Table~\ref{tab:states}.
The dynamics of the helicopter is given in~\cite{mettler2000system}, which can be discretized into a 
linear time-invariant system: 
 $x_{t+1} = Ax_t+ Bu_t$. 
The auto-pilot is supposed to take the helicopter to a waypoint in a 3D-maze within a bounded time $T$  ($\goal$) and avoid the mapped building and trees. 
The complement of these obstacles in the 3D space define the $\safe$ set (see Figure~\ref{fig:heli}).

The computation of the control inputs ($u_t$) typically involves sensing the observable part of the states, computing the inputs to the plant,  and feeding the inputs through actuators. In a cyber-physical system, the mechanisms involved in each of these steps can be attacked and different attacks give rise to different instances of $\ARAS$.

 Controller and Actuator attacks.
An adversary with software privileges  may compromise a part of the controller software.
A network-level adversary may inject spurious packets in the channel between the controller and the actuator. An adversary with hardware access may directly tamper with the actuator and add an input signal of  $a_t$. Under many circumstances, it is reasonable to expect these attacks to be transient or short-lived compared $T$ (for example, otherwise they will be diagnosed and mitigated).
Then the actual input to the system becomes $u'_t = u_t + a_t $ and  the dynamics of the complete system is modified to $x_{t+1} = A x_t + Bu_t + B a_t$, which gives an instance of $\ARAS$.
%

Sensor attacks. Another type of adversary spoofs the helicopter's sensors, the GPS, the gyroscope, so that the position estimator is noisy. Consider a control systems where the adversary-free control $u_t$ is a function on the sequence of sensor data.
If the adversary injects an additive error to the sensors, then the 
control inputs computed based on this inaccurate data will be added an error; also the initial state will have uncertainty. We model the additive error by the adversary input $a_t$.
Once again, this gives rise to an instance of $\ARAS$.
Assuming that the injection of $a_t$ requires energy and that the adversary has limited energy for launching the attack then gives rise the adversary class $Adv = \sum_{i=0}^T ||a_t||^2 \leq b$ where $b$ is the energy budget.



\begin{table}[bth]
\caption{States and inputs of the helicopter model.}
\label{tab:states}
\begin{tabular}{|c|c|}
\hline
States/ Inputs & Description \\
\hline  
$[p_x, p_y, p_z]$ & Cartesian Coordinates \\
$[u, v, w]$ & Cartesian Velocities \\
$[p,q,r]$ & Euler Angular Rates \\
$[a,b,c,d]$ & Flapping Angles  \\
$[\varphi,\phi,\theta]$ & Euler Angles \\
\hline
$u_{z}$ & Lateral Cyclic Deflection in [-1,1] \\
$u_{x}$ & Longitudinal Cyclic Deflection in [-1,1] \\
$u_{p}$ & Pedal Control Input in [-1,1] \\
$u_{c}$ & Collective Control Input in [0,1] \\
\hline
\end{tabular}
\end{table}




\begin{figure}
\centering
\caption{Helicopter fly through scene. Red boxes are the obstacles, the cyan  box on the right is the goal states, the green ball on the left is a set of initial states and the blue curve
is a sampled trajectory of the helicopter with a random adversary input.}
\label{fig:heli}
\includegraphics[width=\textwidth]{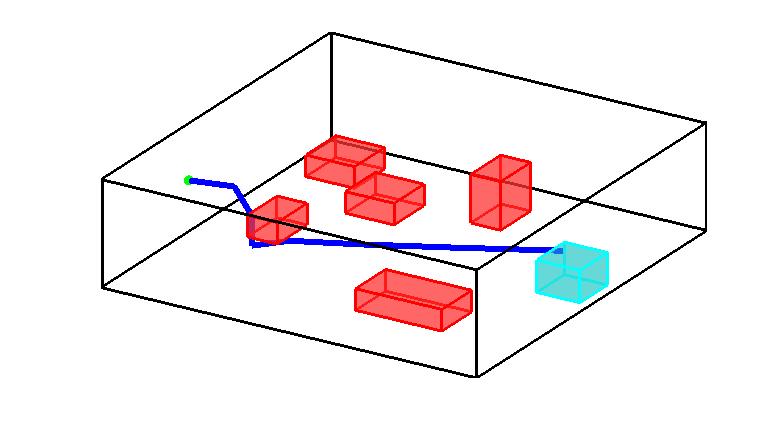}
\end{figure}

%% file: 4decouple.tex
\section{Algorithm for Linear $\ARAS$}

\subsection{Preliminaries and Notations}
\label{sec:notations}

For a natural number $n\in \naturals$, $[n]$ is the set $\{0,1,\dots, n-1\}$.
For a sequence $A$ of objects of any type with $n$ elements, we refer to the $i^{th}$ element, $i \leq n$ by $A_i$.
For a real-valued vector $v \in \reals^n$, $||v||$ is its $\ell^2$-norm. 
For $\delta \geq0$, the set $\B_\delta(v)$ denotes the closed ball $\{ x \in \reals^n \ | \ ||v-x|| \leq \delta\}$
centered at $v$.
%
For a parameter $\epsilon>0$ and a compact set $A\subseteq \reals^n$, an $\epsilon$-cover of $A$ is a finite set $C = \{a_i\}_{i\in I}\subseteq A$ such that $\cup_{i\in I}\B_{\epsilon}(a_i)\supseteq A$.
For two sets $A,B\subseteq \reals^n$, the {\em direct sum} $A\oplus B = \{x\in \reals^n \ : \ \exists a\in A, \exists b\in B, a+b = x\}$. For a vector $v$, we denote $A\oplus v$ as $A\oplus\{v\}$.
Sets in $\reals^n$ will be represented by finite union of balls or polytopes. 
An {\em  $n$-dimensional polytope\/} $P = \{x\in \reals^n \ : \ Ax\leq b\}$ is specified by a matrix $A\in \reals^{m\times n}$ and a vector $b\in \reals^m$,
where $m$ is the number of constraints.
A {\em polytopic set\/} is a finite union of polytopes and is specified by a sequence of matrices 
and vectors. 
A polytopic set can be written in Conjunctive Normal Form (CNF), where
(i)~the complete formula is a conjunction of {\em clauses}, and
(ii)~each clauses is disjunction of linear inequalities.


In this paper, we will assume that the initial set $\init$ is given as a ball $\B_\delta(\theta)\subseteq \X$
for some $\theta \in \X$ and $\delta >0$. We also fix the time horizon $T$. 
The set $Adv$ is specified by a {\em budget} $b\geq 0$:
$Adv = \{\va\in \A^T \ : \ \sum_{t} ||a_t||^2\leq b\}$.
The set $Ctr$ is specified by a polytopic set.

For a sequence of matrices $\{A_t\}_{t\in \naturals}$, 
for any $0\leq t_0 < t_1$,
we denote the
{\em transition matrix} from $t_0$ to $t_1$ inductively as
$\alpha(t_1,t_0)= A_{t_1-1}\alpha(t_1-1,t_0)$ and $\alpha(t_0,t_0) = I$. 

A {\em  trajectory of length $T$\/} for the system is a sequence $x_0, x_1, \ldots, x_T$
such that $x_0 \in \init$ and each $x_{t+1}$ is inductively obtained from Equation~(\ref{eq:dyn}) by 
the application of some admissable controller and adversary inputs.  
The $t^{th}$ state of a trajectory is uniquely defined by the choice of an initial state $x_0\in \init$, 
an admissible control input $\vu \in Ctr$ and an admissible adversary input $\va \in Adv$. 
We denote this state as $\xi(x_0,\vu,\va,t)$.
%

The notion of a trajectory is naturally extended to sets of trajectories with sets of initial states and inputs. 
For a time $t\in[T+1]$, 
a subset of initial states $\Theta\subseteq \init$, 
a subset of adversary inputs ${\bf A} \subseteq Adv$,
and
a subset of controller inputs ${\bf U} \subseteq Ctr$,
we define: 
\[
\reach{}(\Theta,{\bf U},{\bf A},t) = \{\xi(x_0,\vu,\va, t) : x_0\in \Theta \  \wedge  \ \va\in {\bf A\/} \}.
\]
For a singleton $\vu \in \U$, we write $\reach{}(\Theta,\{\vu\},Adv, t)$
as $\reach{}(\Theta,\vu,t)$. To solve $\ARAS$ then we have to decide if 
\begin{equation}
\label{eq:problem}
\begin{array}{c}
\exists \ \vu \in Ctr \ : (\wedge_{t\in [T+1]} \reach{}(\init, \vu, t) \subseteq \safe ) \\
	\wedge \ \reach{}(\init, \vu, T) \subseteq \goal.
\end{array}
\end{equation}
This representation hides the dependence of the $\reach{}$ sets on the set of adversary choices.

%

\subsection{Decoupling}
In this section, we present a technique to decouple the $\ARAS$ problem.
The decomposition relies on a result from robust control that enables us to
precisely compute the reachable states of the system in
terms of a symbolic simulation of the adversary-free dynamics and the
total uncertainty induced by the adversary. In Section~\ref{sec:algomain}, we present an algorithm that
performs this decomposition such as to eliminate the universal quantifier on the adversary's
choices and initial states in Definition~\ref{def:adv} and~\ref{def:init}. 

\subsection{Adversarial Leverage}
\begin{definition}
\label{def:adv} 
For any $t\in[T+1]$, the {\em \advlev} at $t$, initial state $x_0\in \init$, and any control $\vu\in Ctr$,
the {\em \advlev} is a set $R(x_0,\vu,t)$ such that
\begin{equation}
\label{eq:adv}
\reach{}(x_0,\vu,t) = \xi(x_0,\vu, 0,t) \oplus R(x_0,\vu,t)
\end{equation}
\end{definition} 
\noindent 
Informally, the \advlev{} captures how much an adversary can drive the trajectory from an adversary-free trajectory.It decomposes the reach set $\reach{}(x_0,\vu,t)$ into two parts: a deterministic adversary-free trajectory $\xi(x_0,\vu, 0)$, and the reachtube $R(x_0,\vu,t)$ that captures the  nondeterminism introduced by the adversary. %
Our solution for $\ARAS$ heavily relies on computing over-approximations of reach sets and to that end, observe that is suffices to over-approximate \advlev{}.
%
For certain classes of non-linear systems, it can be over-approximated statically using techniques from robust control, such as  $H_\infty$ control. 
It can also be approximated dynamically by reachability algorithms that handle nondeterministic modes (see, for example~\cite{botchkarev2000verification,henzinger2000beyond}). 

For the $\ARAS$ problem with linear dynamics described in~\eqref{eq:dyn}, where the adversary input $Adv = \{\va\in \A^T \ : \ \sum_{t} ||a_t||^2$ $\leq b\}$ is defined by a budget $b\geq 0$, we can compute  \advlev{} precisely.
The following lemma is completely standard in linear control theory.
\begin{lemma}
\label{lem:adv}
For any time $t\in[T+1]$, if the controllability Gramian of the adversary
$W_t=\sum_{s=0}^{t-1}\alpha(t,s+1)C_sC_s^T\alpha^T(t,s+1)$ is invertible, 
then   
\[
R(x_0,\vu,t) = \{x\in \reals^n: x^TW_t^{-1}x\leq b\} 
\]
is the precise \advlev{} at $t$.
\end{lemma}
\begin{proof}
For $t\in [T+1]$, we have
\begin{equation}
x_t = \alpha(t,0)x_0+ \sum_{s=0}^{t-1}\alpha(t,s+1)B_{s}u_{s} + \sum_{s=0}^{t-1}\alpha(t,s+1)C_{s}a_{s}.
\end{equation}
Since $\xi(x_0,\vu,0,t) = A^tx_0+ \sum_{s=0}^{t-1}\alpha(t,s+1)B_{s}u_{s},$ 
we have 
\[
R(x_0,\vu,t) =\{x\in \reals^n: x = \sum_{s=0}^{t-1}\alpha(t,s+1)C_{s}a_{s} \ \wedge \ \sum_{t=0}^{T-1}||a_s||^2 \leq b\},
\]
which is the set $\{x\in \reals^n: x^TW_t^{-1}x\leq b\}$, with controllability Gramian $W_t$.
\end{proof}
The above lemma establishes a precise \advlev{} as an ellipsoid defined by 
the controllability Gramian $W_t$ and $b$. In this case, the ellipsoid is  independent of $x_0$ an $\vu$ and only depends on $t$. Here on, we will drop the arguments of $R$ when they are reduandant or clear from context.
If $W_t$ is singular for some $t\in[T+1]$, then replace the inverse of $W_t$ by its pseudo-inverse and the set $R$ is an ellipsoid in the controllable subspace.

\subsection{Uncertainty in Initial Set}
\label{sec:init}
Following the above discussion, we show that a similar decomposition of the reachable states is possible with respect to the uncertainty in the initial state. 
\begin{definition}
\label{def:init} 
Consider the initial set $\init$ to be $B_\delta(x_0)$ for some $\delta > 0$ and $x_0 \in \X$.
For a $t\in[T+1]$ and a control input $\vu$, the {\em \bloat{}} at $t$ is a set $B(x_0,\vu,t)$, such that 
\begin{equation}
\label{eq:init}
\reach{}(B_\delta(x_0),\vu,0,t) = \xi(x_0,\vu, 0,t) \oplus B(x_0,\vu,t).
\end{equation}
\end{definition} 
The \bloat{} captures the degree to which the uncertainty $\delta$ in the initial set can make the 
adversary-free trajectories deviate.
%
For general nonlinear models, we will have to rely on over-approximating \bloat{}~\cite{}, but for the liner
version of $\ARAS$ the following lemma provides a precise procedure for computing it. 
\begin{lemma}
\label{lem:init}
For an initial set $\init = \B_\delta(\theta)\subseteq \reals^n$, 
for any $t\in[T+1]$, input $\vu \in Ctr$, if the  matrix $\alpha(t,0)^T\alpha(t,0)$ is invertible then 
\[
B(\theta, \vu,t) =\{x\in \reals^n \ : \ x^T [\alpha^T(t,0)\alpha(t,0)]^{-1} x \leq \delta^{1/2}\}
\]
is the precise \bloat{} at $t$.
\end{lemma}
\noindent
If the matrix $A$ is singular, then a similar statement holds in terms of the pseudo-inverse of $[\alpha^T(t,0)\alpha(t,0)]$.
Thus, \bloat{} is an ellipsoid defined by $A,t$ and $\delta$ and is independent of $x_0$ an $\vu$. 
We will drop the arguments of $B$ when they are redundant or clear from context.

\subsection{Adversary-free Constraints}
\label{sec:afc}
Using the  decomposition of the reach set given by the above lemmas, 
we will first solve a new reach-avoid synthesis problem for the adversary-free system.
To construct this new problem we will modify the safety and winning constraints of the $\ARAS$.
For a given time instant, the new constraints are obtained using the same approach as in 
robotic planning with
The synthesis problem requires a solution to a sequence of such problems.  

\begin{definition}
\label{def:afc}
Given a set $S\subseteq \reals^n$ and a compact convex set $R\subseteq \reals^n$,
a set $S' \subseteq \reals^n$  is a {\em \strength{} of $S$ by $R$\/} if
\begin{equation}
\label{eq:strength}
S'\oplus R \subseteq S.
\end{equation}
\end{definition}
A \strength{} $S'$ is {\em precise} if it equals $R \oplus S$. 
The \strength{} $S'$ is a subset of $S$ that is shrunk by the set $R$.  
If $S$ is a polytopic set and $R$ is a convex compact set then exact solutions to the following optimization problem
yields precise \strength{}.

\begin{lemma}
\label{lem:afc}
For a half hyperplane $S = \{x\in\reals^n : c^Tx\leq b\}$ and a convex compact set $R$,
a precise \strength{} of $S$ by $R$ is $S'= \{x\in\reals^n : c^Tx\leq b - c^Tx^*\}$ such
that 
\begin{equation}
\label{eq:opt}
\begin{array}{rc}
x^* = \arg\min\limits_{x \in R} -c^T x.
\end{array}
\end{equation}.
\end{lemma} 
\begin{proof}
Fix any $x\in R$ and $y\in S'$.  
From the definition of $S'$, $c^Ty + b^* \leq b$.
Since  $x^*$ minimizes $-c^T x$ in $R$ and $x\in R$,
we have $-c^T x \geq -c^T x^* =  b^*$.
It follows that $c^T(x+y) \leq c^T y + c^T x^* \leq c^Ty + b^* \leq b$.
Thus $x+y \in S$ and therefore $S'\oplus R \subseteq S$.

For any $y\in S$, it holds that $c^T y \leq b$. 
Let $y' = y - x^*$.
It follows that $c^T y' = c^T y - c^Tx^*\leq b-c^T x^*$.
Thus $y'\in S'$. Combined with $x^*\in R$, $y = y'+x^*\in S'\oplus R$. 
Therefore $S'\oplus R \subseteq S$.
%
\end{proof}
Since a polytopic set is a union of intersections of linear inequalities,
the above lemma generalizes to polytopic sets in natural way. 
\begin{corollary}
\label{coro:afc}
For a polytopic set $S =\{x\in \reals^n \ : \ \vee_{i\in[m]} A_i x \leq b_i  \}$
and a compact convex set $R\subseteq \reals^n$, 
\[
S'  = \{x\in \reals^n \ : \ \bigvee_{i\in[m]} A_i x \leq b_i - b_i^*  \},
\]
is a precise \strength{} of $S$ by $R$. Here the $j^{th}$ element of $b_i^*$ equals $c^Tx^*$
with $c^T$ being the $j^{th}$ row of $A_i$ and $x^*$ is the solution of~\eqref{eq:opt}.
\end{corollary}

%% file: 5algorithm.tex
\subsection{An Algorithm for Linear $\ARAS$}
\label{sec:algomain}
We present algorithm~\ref{alg:precise} for solving the linear version of the $\ARAS{}$ problem.

\begin{algorithm}
\caption{$Synthesis(\init,\safe,\goal,Adv, Ctr, T)$}
\label{alg:precise}
\SetKwInOut{Input}{input}
\SetKwInOut{Output}{output}
	\For{ $t\in [T+1]$}{ \label{ln:1}
		$R_t \gets \mathit{AdvDrift}(Adv,t)$\;\label{ln:adv}
		$B_t \gets \mathit{InitCover}(\init,t)$\;\label{ln:init}
		$\safe'_t \gets \mathit{Strengthen}(\safe,R_t,B_t)$\;\label{ln:str}
	}
	$\goal' \gets \mathit{Strengthen}(\goal,R_T,B_T)$\; \label{ln:goal}
	$(\vu, \mathit{Failed}) \gets \mathit{SolveSMT}(\theta, \safe',\goal',Ctr, T )$\;\label{ln:smt}
	\Return $(\vu, \mathit{Failed})$
\end{algorithm}

The subroutine $\mathit{AdvDrift}$ computes a precise \advlev{} $R_t$ for every time $t\in [T+1]$. 
From Lemma~\ref{lem:adv}, $R_t$ is an ellipsoid represented by the controllability Gramian and the constant $b$. 
The subroutine $\mathit{InitCover}$ computes a \bloat{} described in Lemma~\ref{lem:init} for each $t$.
The subroutine $\mathit{Strengthen}$ computes a precise strengthening of the safety constraints $\safe$ by both sets  $R_t$ and  $B_t$.
From Corollary~\ref{coro:afc}, the strengthening is computed by solving a sequence of optimization problems.
Since $R_t$ and $B_t$ are both ellipsoids (Lemma~\ref{lem:adv} and~\ref{lem:init}), 
the optimization problems solved by $\mathit{Strengthen}$ are quadratically constrained linear optimization problems and are 
solved efficiently by second-order cone programming~\cite{alizadeh2003second} or semidefinite programming~\cite{vandenberghe1996semidefinite}.
For each $t \in [T+1]$, the  set $\safe$ is strengthened by the corresponding adversary drift $R_t$ to get $\safe'_t$. 
The $\goal$ set is  strengthened respect to the adversary drift at the final time $T$ to get $\goal'$.
Finally, $\mathit{SolveSMT}$ makes a call to an SMT solver to check if there exists a satisfiable assignment $\vu\in Ctr$ for quantifier-free 
formula~\eqref{eq:SMT}:  
\begin{equation}
\label{eq:SMT}
\begin{array}{c} 
\exists \ \vu \ \in Ctr \  \wedge \\ 
(\wedge_{t\in [T+1]} \xi(\theta,\vu,0,t)\in \safe'_t) \ \wedge \ \xi(\theta,\vu,0,T)\in \goal'.
\end{array}
\end{equation}
For the class of problems we generate, the SMT solver terminates and either returns a satisfying assignment $\vu$  
or it proclaims the problem is unsatisfiable by returning $\mathit{Failed}$.
If $\mathit{AdvDrift}$, $\mathit{InitCover}$ and $\mathit{Strengthen}$ compute \advlev, \bloat{} and \strength{} precisely,
then Algorithm~\ref{alg:precise} is a sound and complete for the linear $\ARAS$ problem.

\begin{theorem}
\label{thm:main}
Algorithm~\ref{alg:precise} outputs $\vu\in Ctr$ if and only if $\vu$ solves $\ARAS$.
\end{theorem}
\begin{proof}
Suppose Algorithm returns $\vu \in Ctr$. We will first show that $\vu$ solves $\ARAS$.
Since $\vu$ satisfies constraints~\eqref{eq:SMT}, for every $t\in[T+1]$, $\xi(\theta,\vu,0,t)\in S_t$.
Since $S_t$ is a strengthening of $\safe$ by $R_t$ and $B_t$, we have $S_t \oplus R_t \oplus B_t \subseteq \safe$.
Thus, 
\begin{equation}
\label{eq:alg1}
\xi(\theta,\vu,0,t)\oplus S_t \oplus B_t \subseteq \safe.
\end{equation}
By Definition~\ref{def:adv} and~\ref{def:init}, we have 
\begin{equation}
\label{eq:alg2}
\begin{array}{rcl}
\xi(\theta,\vu,0,t)\oplus \safe'_t \oplus B_t &\supseteq& \reach{}(\theta,\vu,Adv,t)\oplus B_t \\
 &\supseteq& \reach{}(\init,\vu,t).
\end{array}
\end{equation} 
Combining~\eqref{eq:alg1} and~\eqref{eq:alg2}, we have 
$\reach{}(\init,\vu,t) \subseteq \safe$. That is  the safety condition of~\eqref{eq:problem} holds.
Similarly, since $\goal'$ is the strengthening of $\goal$ by $R_T$ and $B_T$, we have $\reach{}(\init,\vu,T) \subseteq \goal$. 
The winning condition also holds. 

On the other side, suppose $\vu\in Ctr$ solves $\ARAS$, it satisfies~\eqref{eq:problem}.  
Since the \advlev{} $R_t$, \bloat{} $B_t$ and \strength{} $\safe',\goal'$ are computed precisely,
Equations~\eqref{eq:alg1} and~\eqref{eq:alg2} take equality.
Thus, for any $t\in [T+1]$, $\xi(\theta,\vu,0,t)\in \safe'_t$ and $\xi(\theta,\vu,0,T)\in \goal'$.
Therefore $\vu$ is returned by Algorithm~\ref{alg:precise}. 
\end{proof}

The completeness of the algorithm is based on two facts:
(i) \advlev{}, \bloat{} and \strength{} can be computed precisely,
and (ii) the SMT solver is complete for formula~\eqref{eq:SMT}.
The exact computation of \advlev{} and \bloat{} require that 
the initial state $\init$ and admissible adversary $Adv$ are described by $\ell^2$ balls.
Since $Ctr$, $\safe'$ and $\goal'$ are polytopic sets, formula~\eqref{eq:SMT}
is a quantifier-free theory in linear arithmetic, which can be solved efficiently for example
by algorithm DPLL(T)~\cite{dutertre2006fast}.

%% file: 6general.tex
\section{Generalizations}
\label{sec:general}
In this section, we discuss two orthogonal generalizations of linear $\ARAS$ and algorithms for solving them building on the algorithm $Synthesis$. First in Section~\ref{sec:convex},
we present an approximate approach to solve a problem where
$\init$, $Adv$ and $Ctr$ are general compact convex sets.
Then, in Section~\ref{sec:state-depend}, we modified the definition 
of linear $\ARAS$ problem such that
the controller  can be a function of the initial states.
A solution of this problem is a look-up table, where the controller choose a sequence of 
open loop control depending on the initial state. 

\subsection{Synthesis for Generalized Sets}
\label{sec:convex}
We generalize the linear $\ARAS$ problem described in Section~\ref{sec:notations} such that 
$\init\subseteq \X$, $Ctr\subseteq \U^T$ and $Adv\subseteq \A^T$ 
are assumed to be some compact subsets of Euclidean space.
For a precision parameter $\epsilon> 0$, the generalized $\ARAS$
problem can be approximated by a linear $\ARAS$ problem. 
We define robustness of a $\ARAS$ problem.

We present an extension of $\mathit{Synthesis}$ to solve this problem. 
For a parameter $\epsilon> 0$, and compact convex sets $\init,Adc,Ctr$,
we construct a tuple $(\Theta,\bf{A}, \C)$ 
such that
\begin{enumerate}[(i)]
\item $\Theta = \{\theta_i\}_{i\in I}$ is an $\epsilon$-cover of initial set $\init$, that is, $\init \subseteq \cup_i B_\epsilon(\theta_i)$. 
\item $\bf{A} = \{\va_j\}_{j\in J}$ is an $\epsilon$-cover of the adversary.   
Here each $\va_j$ is seen as a vector in Euclidean space $\A^T$ and the union of $\epsilon$-balls around each $\va_j$ over-approximates $Adv$.

%
\item $\C \subseteq Ctr \subseteq \U^T$ is a polytopic set such that $d_H(\C, Ctr) \leq \epsilon$,
That is, $\C$ under-approximates the actual constraints of control $Ctr$, with error bounded by $\epsilon$ measured by Hausdorff distance. 
\end{enumerate}

The modified algorithm to approximately solve the generalized $\ARAS$ problem follows the same steps as Algorithm~\ref{alg:precise} from line~\ref{ln:1} to line~\ref{ln:goal}.
The only change is in  line~\ref{ln:smt}, where instead of solving an SMT formula~\eqref{eq:SMT}
we solve~\eqref{eq:SMTapprox}. 
\begin{equation}
\label{eq:SMTapprox}
\begin{array}{c} 
\exists \ \vu \ : \  \vu  \in \C \  \wedge \\ 
(\wedge_{t\in [T+1]} \wedge_{i\in I} \wedge_{j\in J} \xi(\theta_i,\vu,\va_j,t)\in \safe'_t )\ \wedge \\
(\wedge_{i\in I} \wedge_{j\in J}  \xi(\theta,\vu_i,\va_j,T)\in \goal')
\end{array}
\end{equation}

The soundness of this modified algorithm is independent of the choice of $\epsilon>0$. 
That is, if it returns a satisfiable assignment $\vu$, then $\vu$ solves the $\ARAS$ problem. 
\begin{lemma}
\label{thm:approx}
If the modified algorithm returns $\vu\in \C$, then $\vu$ solves linear generalized $\ARAS$. 
\end{lemma}
\begin{proof}
Suppose $\vu\in \C\subseteq Ctr$ satisfies~\eqref{eq:SMTapprox}. 
Since $\Theta$ and $\bf{A}$ are $\epsilon$-cover of $\init$ and $Adv$, 
there exist a initial state $\theta_i\in $
for any $t\in [T+1]$ we have
\[
\reach{}(\init,\vu,Adv,t) \subseteq \reach{}(\cup_{i\in I}\B_\epsilon(\theta_i),\vu,\cup_{j\in J}\B_\epsilon(\va_j) ,t).
\]
Let $R_t$ and $B_t$ be the precise \advlev{} and \bloat{} as in Algorithm~\ref{alg:precise}.
From Lemma~\ref{lem:adv} and~\ref{lem:init}, $R_t$ and $B_t$ are independent on the initial state and adversary input.  
Therefore, 
\begin{equation}
\label{eq:partreach}
\begin{array}{rl}
 &\reach{}(\init,\vu,Adv,t)\\
 =& \cup_{i\in I}\cup_{j\in J} \reach{}(\B_\epsilon(\theta_i),\vu,\B_\epsilon(\va_j),t)\\
 =&  \cup_{i\in I}\cup_{j\in J} ( \xi(\theta_i,\vu,\va_j,t) \oplus R_t \oplus B_t ) \\
 =&  (\cup_{i\in I}\cup_{j\in J}  \xi(\theta_i,\vu,\va_j,t)) \oplus R_t \oplus B_t .
\end{array}
\end{equation}
From formula~\eqref{eq:SMTapprox} implies that $(\cup_{i\in I}\cup_{j\in J}  \xi(\theta_i,\vu,\va_j,t)) \subseteq \safe_t'$ for any $t\in [T+1]$ and $(\cup_{i\in I}\cup_{j\in J}  \xi(\theta_i,\vu,\va_j,T))\subseteq \goal'$.
Since $\safe'_t$ is an $R_t\oplus B_t$ \strength{} of $\safe$, it follows from Definition~\ref{def:afc} and~\eqref{eq:partreach} that
$\reach{}(\init,\vu,Adv,t) \subseteq \safe$ for all $t\in[T+1]$ and $\reach{}(\init,\vu,Adv,T) \subseteq \goal$.
That is, $\vu$ solves the generalized linear $\ARAS$.
\end{proof}
We observe that if the approximated algorithm successfully synthesize a control, the control solves the generalized linear $\ARAS$ problem, no matter what value $\epsilon>0$ takes.
Moreover, as the parameter $\epsilon$ converges to 0, 
we have $\cup_{i\in I}\B_\epsilon(\theta_i)$,  $\cup_{j\in J}\B_\epsilon(\va_j)$ and $\C$ converge to
the exact $\init$, $Adv$ and $Ctr$, respectively.

\subsection{State-dependent Control}
\label{sec:state-depend}
In this section, we keep the same definition of $\init$, $Adv$ and $Ctr$ as 
in Section~\ref{sec:notations}, however, we consider a variant of $\ARAS$ that allows the choice of control $\vu$ to be depend on the initial state of the system. 
That is, we have to decide if
\begin{equation}
\label{eq:problem_init}
\begin{array}{c}
\forall \ x_0\in \init\ : \ \exists \ \vu \in Ctr \ : \\
 (\wedge_{t\in [T+1]} \reach{}(x_0, \vu, t) \subseteq \safe )
	\wedge \ \reach{}(x_0, \vu, T) \subseteq \goal.
\end{array}
\end{equation}
A solution to this generalized $\ARAS$ problem is a look-up table $\{(\I_i,\vu_i)\}_{i\in I}$ 
such that 
\begin{inparaenum}[(i)]
\item the union $\cup_{i\in I}\I_i \supseteq \init$ covers the initial set, and
\item for every $x_0\in \I_i$, $\vu_i$ is an admissible input such that the constraints in~\eqref{eq:problem_init} hold.
\end{inparaenum}

We present an Algorithm~\ref{alg:approx} to solve this problem and it uses $Synthesis$ as an subroutine. If the algorithm succeeds,  it returns a look-up table Tab which solves the above state-dependent variant of  $\ARAS$.

The parameters $Adv,Ctr,\safe,\goal,T$ are invariant in the algorithm, thus we omit it as arguments of $Synthesis$.
The variable $\epsilon$ is initialized as the diameter of the initial set $\init$ (line~\ref{ln:2dia}).
The subroutine Cover($\init,\epsilon$) in line~\label{ln:2approx} first computes an $\epsilon$-cover $\{\theta_i\}_{i\in I}$ of $\init$, 
and then append each $\theta_i$ with the parameter $\epsilon$.
The set $\S$ stores all such pairs $(\theta,\epsilon)$, 
such that the $\epsilon$-ball around $\theta$ is yet to examined by the algorithm for $Synthesis$.
For each ball $\B_\epsilon(\theta)$ in $\S$, the subroutine $\mathit{Synthesis}$ is possibly called twice for
both the ball $\B_\epsilon(\theta)$ and the single initial state $\theta$
 to decide whether the $Synthesis$ is  successful, a failure, or whether further refinement is needed.


\begin{algorithm}
\caption{$TableSynthesis$}
\label{alg:approx}
	$\epsilon\gets $Dia$(\init)$\; \label{ln:2dia}
	$\S \gets $Cover$(\init,\epsilon)$\; \label{ln:2approx}
	Tab$ \gets \emptyset$\;
	\While{ $\S \neq  \emptyset$ {\bf For} $(\theta,\epsilon)\in \S$ }{ \label{ln:while}
		$\S \gets \S/\{(\theta,\epsilon)\}$\;
		\uIf{ $Synthesis(\B_\epsilon(\theta))$ returns $\vu\in Ctr$ }{ \label{ln:over}
		Tab$ \gets $ Tab$ \cup \{(\B_\epsilon(\theta),\vu)\}$\; \label{ln:2unsat}
		}
		\uElseIf{$Synthesis(\B_\epsilon(\theta))$ failed }{ \label{ln:under}
		\Return ($\theta$,Failed)
		}
		\Else{		
		$\S \gets \S \cup \mbox{Cover}(\init\cap \B_\epsilon(\theta),\epsilon/2)$\;
		}
	}
	\Return (Tab, Success)
\end{algorithm}

\begin{theorem}
\label{thm:sound}
If {\em TableSynthesis} returns {\em (Tab,Success)}, then {\em Tab} solves the state-dependent $\ARAS$.
Otherwise if {\em Tablesynthesis} returns {\em ($\theta$,Failed)}, then there is no solution for initial state $\theta$.
\end{theorem}
\begin{proof}
We first state an invariant of the while loop which can be proved straightforwardly through induction. 
For any iteration, suppose Tab$= \{(\B_{\epsilon_i}(\theta_i),\vu_i)\}_{i\in I}$ and $\S = \{\theta'_j,\epsilon'_j\}_{j\in J}$ are the valuations of $Tab$ and $\S$ at the beginning of the iteration. 
Then we have $(\cup_{i\in i}\B_{\epsilon_i}(\theta_i)) \cup ( \cup_{j\in J} \B_{\epsilon'_j}(\theta'_j)) \supseteq \init$.

Suppose $TableSynthesis$ returns (Tab,Success) with Tab$=\{(\B_{\epsilon_i}(\theta_i),\vu_i)\}_{i\in I}$. From line~\ref{ln:while}, $\S = \emptyset$.
From the loop invariant, we have $\cup_{i\in i}\B_{\epsilon_i}(\theta_i)\supseteq \init$.
Moreover for any $(\B_{\epsilon}(\theta),\vu)\in$Tab, from line~\ref{ln:over} and Theorem~\ref{thm:main},
for any $x_0\in \B_{\epsilon}(\theta)$, $\vu$ is an admissible input such that constraints in~\ref{eq:problem:init} hold. Thus Tab solves the state-dependent $\ARAS$.

Otherwise suppose $TableSynthesis$ returns ($\theta$,Failed). From line~\ref{ln:under} and Theorem~\ref{thm:main}, there is no admissible $\vu$ solve the $\ARAS$ from $\theta$.
\end{proof}

The Algorithm~\ref{alg:approx} is sound, that is, if the algorithm terminates, it always returns the right answer.
For general sets of $Adv$ and $Ctr$ the approach from Section~\ref{sec:convex} can be combined Algorithm~\ref{alg:approx} to get state dependent (but $\vu$ and $\va$ oblivious) controllers.

%% file: 7experiment.tex
\section{Implementation and Experimental Evaluation}
We have implemented the algorithm $\mathit{Synthesis}$ in a prototype tool in Python.
The optimization problem presented in Lemma~\ref{lem:afc} 
is solved by a second-order cone programming solver provided by package CVXOPT~\cite{dahl2006cvxopt}.
The quantifier-free SMT formula~\eqref{eq:SMT} is solved
by Z3 solver~\cite{de2008z3}.
In Section~\ref{sec:expsynth} and \ref{sec:complexity},
we present the implementation of the basic algorithm 
$\mathit{synthesis}$, show an example in detail, 
present the experiment results
and discuss the complexity of the algorithm.
In Section~\ref{sec:vulnerable} and~\ref{sec:advsyn}, 
we present several different applications of $Synthesis$.

\subsection{Synthesizing Adversary Resistant Controllers}
\label{sec:expsynth}
We have solved several linear $\ARAS$ problems for a 16-dimensional helicopter system (as described in~\ref{sec:heli}) and a 4-dimensional vehicle. 

We illustrate an instance of the synthesis of the helicopter auto-pilot for time bound $T=9$ in Figure~\ref{fig:sim}.
The state variables, control input variables and the constraint $Ctr$ of the system are listed in Table~\ref{tab:states}.
We model an actuator intrusion attack such that
the control input is tempered by an amount of $a_t$ at each time $t\in [T]$.
The total amount of spoofing is bounded by a budget $b=1$. 

A control $\vu = \{u_t\}_{t\in [T]}$ is synthesized by $Synthesis$.
We randomly sample  adversary inputs $\va$ with  
$\sum_{t\in [T]}||a_t||^2 = b$, 
and visualize the corresponding trajectories with control $\vu$ in Figure~\ref{fig:sim}.

\begin{figure}
\centering
\includegraphics[width=\textwidth]{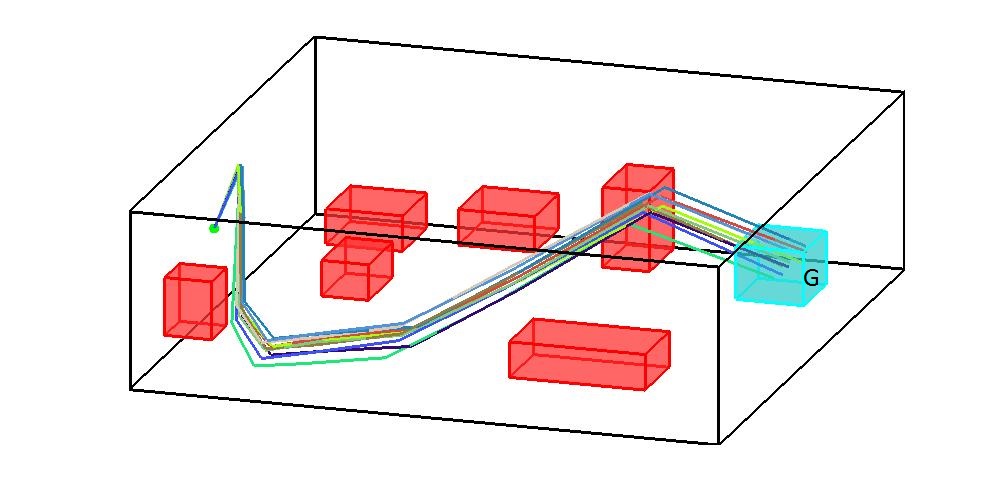}
\caption{Sampled Trajectories of Helicopter Auto-pilot. Safety and winning conditions hold.}
\label{fig:sim}
\end{figure}

Besides the Helicopter model, we studied an discrete variation of the navigation problem
of a 4-dimensional vehicle, where the states are positions and velocities in Cartesian coordinates,
and the controller and adversary compete to decide accelerations in both direction.

The experimental results for different instances are listed in Table~\ref{tab:main}, where the columns 
represent (i) the model of the complete system, (ii) the dimension of state, control input and adversary input vectors,
(iii) the time bound, (iv) the length of formula representing $\safe$ and number of obstacles,
(v) the length of formula representing $\goal$ and $Ctr$,
(vi) the length of the quantifier-free formula in~\eqref{eq:problem},
(vii) the synthesis result, and (vii) the running time of the synthesis algorithm.

From the result, we observe that the algorithm can synthesize controller 
for lower dimensional system for a relatively long horizon ($320$) for reasonable amount of time.
For higher dimensional system (16-dimensional), the approach scales to an horizon $T=15$.
The run time of the algorithm grows exponentially with the time bound $T$.
By Comparing row 2-4, we observe that the runtime grows linearly with the number of obstacles.

\begin{table*}[tbh]
\caption{Experimental results for $Synthesis$}
\label{tab:main}
\begin{tabular}{c||c|c|c|c|c|c|c}
Complete System 	&  \# $x,u,a$  & $T$ &  $|\phi_\safe|, \# Obs $& $|\phi_\goal|,|\phi_{Ctr}| $ & $|\phi|$   & Result &Run Time (s) \\
\hline
\multirow{6}{*}{Vehicle}  &  \multirow{6}{*}{4,2,2}      & 40   &  16,  3 &4, 160 &  804 &  unsat	&      2.79  \\
				    &  	   				     & 80   &   20,  4  &4, 320 & 1924  &  sat	&      16.49  \\
				    &  	   				     & 80   &  44,  10 &4, 320 & 3844  &  sat	&      35.22    \\
				    &  	   				     & 80   &  84,  20 &4, 320 & 7044  &  sat	&      53.8    \\ 
				    &  	   				     & 160  &  20, 5  &4, 640 & 3844  &  sat		&      91.78  \\
				    &  	   				     & 320  &  24, 6  &4, 1280 & 8964  &  sat	&      532.5  \\
\hline
\multirow{7}{*}{Helicopter}&\multirow{7}{*}{16,4,4}    &  5   &  18, 3 &6,40  & 136 &sat & 1.2 \\ 
			 &      				   	        & 5   & 24, 4  &6,40 &  166  & unsat &  0.61 \\  
			  &      				     	& 7   & 24, 4  & 9, 56 &  213  & sat &  8.2 \\  
			  &      				     	& 9   & 36, 6    &  6,  72 &402  & sat &  24.5 \\  
			 &      				     	& 12    & 24,4  &6, 96, & 338 & sat &  60.6   \\  
			 &      				     	& 15    & 24, 4  &6, 96, &576 & sat &  158.8   \\  	
			 &      				     	& 18    & 24, 4&10, 96, &640 & -- &  --   \\  		 

\end{tabular}
\end{table*}

\subsection{Discussion on Complexity of Safety Constraints}
\label{sec:complexity}

Let the quantifier-free constraints in~\eqref{eq:problem} be specified by an CNF formula $\phi$,
where each atomic proposition is a linear constrain. 
We denote $|\phi|$ as the length of the CNF formula which is the number of atomic propositions in $\phi$.
Notice that if we convert an CNF formula into a form of union of polytopes,
the size of the formula can grow exponentially.
Similarly, let CNF formula $\phi_\safe$, $\phi_\goal$ and $\phi_{Ctr}$ specify the constraints $\safe,\goal \subseteq \X$ and $Ctr\subseteq \U^T$.
It can be derived from~\eqref{eq:problem} that $|\phi| = T|\phi_\safe|+|\phi_\goal| + |\phi_{Ctr}|$.
If fixed the length of the projection of $\phi_{Ctr}$ on
control $u_t$ for each $t$, that is, we assume the controller constraints at different times are comparably complex, 
then $|\phi_{Ctr}|$ grows linear with the time bound $T$.
Suppose the length of $|\phi_\safe|,|\phi_\goal|$ are constant, then the length of $\phi$ is linear to the time bound $T$.

The length of $\phi_\safe$ is a function of the number and complexity of obstacles.
Suppose that the safe region $\safe'$ is obtained by adding an polytopic 
obstacle $O = \{x\in \reals^n : A x < b  \}$
to a safe region $\safe$.  One measure of complexity of the obstacle is the number of rows of the matrix $A$.
Then, the resulting safe region is $\safe' = \safe \backslash O$, which implies
\[
\phi_{\safe'} = \phi_\safe \wedge \neg(Ax< b) = 
\phi_\safe \wedge (\vee_{i} -A_ix \leq -b),
\]
where $A_i$ is the $i^{th}$ row of $A$.
Therefore the length of $\phi_\safe$ increases linearly with the number of obstacles
and the number of faces in every obstacle.

In the experiments, we observe that the running time of Z3 to solve the SMT formula varies on a case by case basis. The size of obstacles, the volume of the obstacle-free region and the length of 
significant digits of entries the constraints and dynamic matrices also affect the running time.  


\subsection{ Vulnerability Analysis of Initial States}
\label{sec:vulnerable}
Using $Synthesis$, we can examine 
the vulnerability of initial states to attackers. 
Fixing a controller constraint $Ctr$, a time bound $T$,  safety condition $\safe$ and winning condition $\goal$,
for each initial state $\init$, there exists a maximum critical budget $b_\mathit{mfc}$ of the adversary $Adv$, such that
beyond this budget, the problem becomes infeasible.
The lower the $b_\mathit{mfc}$ for an initial state is, it is vulnerable to a weaker adversary.
The maximum budget can be found by a binary search on the adversary budget with {\em Synthesis}.

We examine the vulnerability of an instance of the 4-dimensional autonomous vehicle system. The result is illustrated in Figure~\ref{fig:budbnd},
where the box at the bottom represent the $\goal$, the red regions represent the obstacle whose complement is the $\safe$, the green-black on the top region is the $\init$.
The black regions are most vulnerable with $b_\mathit{mfc} = 0$ and the lightest green region
are least vulnerable with $b_\mathit{mcc} = 1.8$. 
We see that the region  closer to an obstacle are darker as an adversary with relatively small budget ($b$) 
can make the vehicle run into an obstacle. 
We also observe that the dark regions are shifted towards the center since the obstacles are
aggregated at the center of the plane. Avoiding them may cause a controller run out of the time bound. 

\begin{figure}
\centering
\caption{Vulnerability Analysis of Initial States. Adversary may cause the system to hit an obstacle or delay the time of reaching beyond $T$}
\label{fig:budbnd}
\includegraphics[width=\textwidth]{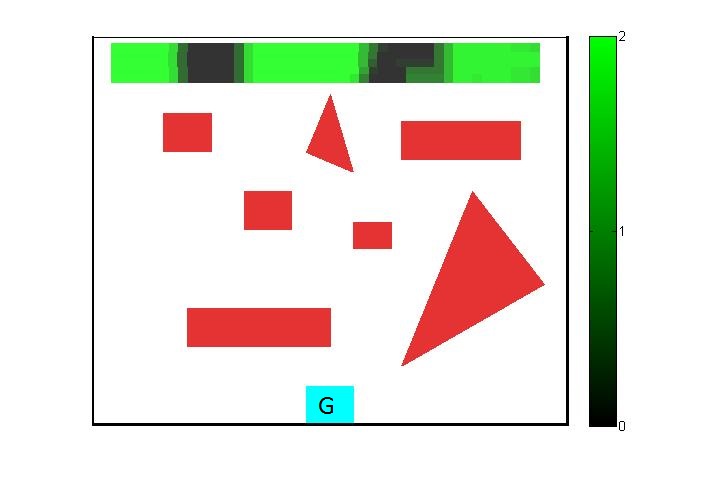}
\end{figure}

\subsection{Attack Synthesis}
\label{sec:advsyn}
The $Synthesis$ subroutine can also be used to generate attacks by swathing the roles of 
the adversary and the controller.
In this section, we  synthesize adversarial attacks to the 4-dimensional vehicle such that 
the system will be driven to unsafe states in a bounded time $T$. 
That is, for a state $x\in \X$, we decide whether  
\begin{equation}
\label{eq:synadv}
\begin{array}{c}
\exists \ \va \in Adv \ \forall \ \vu\in Ctr : \\
(\vee_{t\in [T]} \xi(x, \vu,\va, t) \in \unsafe ).
\end{array}
\end{equation}
Notice that~\eqref{eq:synadv} is essentially the same as~\eqref{eq:problem} by switching 
the roles of $\vu$ and $\va$, and negating $\safe$ to get $\unsafe$.

We suppose that the set of adversarial input $Adv$ is a polytopic set  and the control 
$Ctr = \{\vu\in \U^T\ : \ \sum_{t\in T} ||u_t||^2 \leq b \}$
is specified by budget $b\geq 0$. 
For general convex compact sets $Ctr$ and $Adv$,
one can come up with an under approximated  $Adv$ as polytopic set and
an over-approximated $Ctr$ with budget $b$.
As we discuss in Section~\ref{sec:convex}, this approximation is sound.

We synthesize a look-up table $\{(\I_i,\va_i)\}_i$ as the  strategy of the adversary,
such that
(i) $\I_i \subseteq \X$, and 
(ii) for each state $x\in \I_i$, the corresponding adversary $\va_i$
satisfies~\eqref{eq:synadv}.
During the evolution of the plant under controller, 
the adversary act only when the system reaches a state  $x\in \I_i$ for some $\I_i$ in the look-up table,
then the corresponding attack  $\va_i$ is triggered at $x$ which breaks the safety of the system.

The synthesis of attacks uses similar idea of creating covers of the states as in $TableSynthesis$ without refinements.
Suppose the set of states $\X\subseteq \reals^4$ is compact. An adversary first creates a uniform cover of the state space,
then search for an attack for each cover. If the synthesis succeed and returns an attack $\va$,
then the cover is {\em vulnerable} and is stored in the look-up table of attacks paired with the attack $\va$.
 
A result of the synthesis is illustrated in~\ref{fig:attack}, where the red boxes specify obstacles. 
The vulnerable covers, each of which is a subset of $\reals^4$,  are projected on the 2-D plane and visualized 
as blue regions, where the white region are not vulnerable to attackers.
The darkness of a region corresponds to the number of vulnerable covers have projection in the region.
That is, if the vehicle is in a dark region, a large portion of its velocity space is vulnerable under attacks that makes the system unsafe.
A sample trajectory is captured by the green curve, where, as it enters light shadow region, its velocity does not fall into a vulnerable cover right away. As it approach further, it enters a vulnerable cover and an attack is triggered
at the point with cross mark.

\begin{figure}
\centering
\caption{Attack Generation. The darker a region is, a larger portion of velocity is vulnerable. If the vehicle visit a region near to an obstacle, it could survive only if its initial velocity is pointing outwards.}
\label{fig:attack}
\includegraphics[width = .7\textwidth]{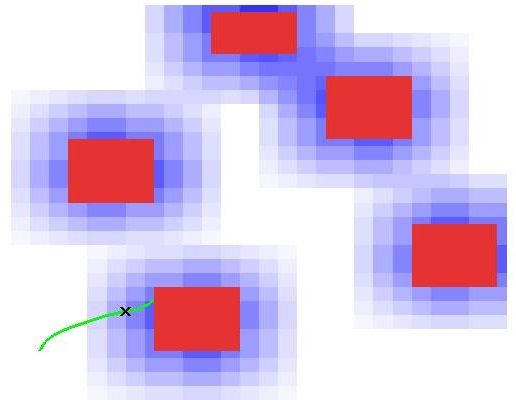}
\end{figure}

%% file: 8conclusion.tex
\section{Conclusion}
\label{sec:conclusion}

We present a controller synthesis algorithm for a discrete time
reach-avoid problem in the presence of adversaries. 
Specifically, we present
a sound and complete algorithm for the case with linear time-varying dynamics and an
adversary with a budget on the total L2-norm of its actions. The
algorithm combines techniques in control theory and synthesis approaches coming 
from formal method and programming language researches.
Our approach first precisely converts the reach set of the complete system
into a composition of non-determinism 
from the adversary input and the choice of initial state,
and an adversary-free trajectory with fixed initial state. 
Then we enhance the $\safe$ and $\goal$ conditions
by solving a  sequence of quadratic-constrained linear optimization problem.
And finally we derive a linear quantifier-free SMT formula
for the adversary-free trajectories, which can be solved effectively by SMT solvers.
The algorithm is then extended to solve  problems 
with more general initial set and constraints of controller and adversary.
We present
preliminary experimental results that show the effectiveness of this
approach on several example problems. 
The algorithm synthesizes adversary-resilient controls 
for a 4-dimensional system for 320 rounds and for  
a 16-dimensional system for 15 rounds in minutes.
The algorithm is extended to analyze vulnerability of 
states and to synthesize attacks.

\subsection*{Future Direction}

There are several interesting follow-up research topics.
For example, the solution of linear $\ARAS$ can be used to solve adversary-free 
nonlinear avoid-reach problems, 
where the dynamics can be
linearized along a nominal trajectory and the linearization error
is modeled as adversary.

We also planned to extend the 
approach to 
synthesize switched controller for infinite horizon
by applying a similar approach as suggested in~\cite{leeDetect}.

Another interesting direction
is to precisely define a dual problem of the linear $\ARAS$.
Since reachability is dual to detectability,
we envision that there exists a detectability type problem dual to $\ARAS$,
such that the adversary adds noise to the measurements. The question is
then how well we can estimate whether the system is in unsafe state
 based on the noisy measurements.